\documentclass[journal]{IEEEtran}
\usepackage{cite}
\usepackage{amsmath,amssymb,amsfonts}
\usepackage[ruled,vlined,linesnumbered]{algorithm2e}
\usepackage{graphicx}
\usepackage{textcomp}

\usepackage{amsmath,amsthm,amssymb,paralist,subfigure,cite,graphicx,color}
\newtheorem{lem}{Lemma}

\newtheorem{theorem}{Theorem}

\newtheorem{rem}{Remark}

\def\mb{\mathbf}

\def\mc{\mathcal}

\newcommand{\ab}[1]{\textcolor{blue}{{#1}}}

\begin{document}
	\title{Accelerated Distributed Allocation}
	
	\author{Mohammadreza Doostmohammadian$^\ast$, Alireza Aghasi$^\dagger$
		\thanks{		
			$^\ast$ Faculty of Mechanical Engineering, Semnan University, Semnan, Iran {\texttt{doost@semnan.ac.ir}}.
			$^\dagger$ Department of Electrical Engineering and Computer Science, Oregon State University, Oregon, USA  {\texttt{alireza.aghasi@oregonstate.edu}}.}}
	\maketitle
	\thispagestyle{empty} 
	\begin{abstract}
		Distributed allocation finds applications in many scenarios including CPU scheduling, distributed energy resource management, and networked coverage control. In this paper, we propose a fast convergent optimization algorithm with a tunable rate using the signum function. The convergence rate of the proposed algorithm can be managed by changing two parameters. We prove convergence over uniformly-connected multi-agent networks. Therefore, the solution converges even if the network loses connectivity at some finite time intervals. The proposed algorithm is all-time feasible, implying that at any termination time of the algorithm, the resource-demand feasibility holds. This is in contrast to asymptotic feasibility in many dual formulation solutions (e.g., ADMM) that meet resource-demand feasibility over time and asymptotically.
	\end{abstract}
	\begin{IEEEkeywords}
		Distributed allocation, scheduling, graph theory, optimization
	\end{IEEEkeywords}
	
	\section{Introduction} \label{sec_intro}
	\IEEEPARstart{A}{llocation} and scheduling find applications in  CPU scheduling  \cite{ojsys,ccta2,grammenos2023cpu}, energy resource allocation  \cite{cherukuri2015distributed,derm,kar_edp,alizadeh2012demand}, linearly constrained minimum variance (LCMV) beamforming \cite{zhang2019distributed}, coverage control \cite{MiadCons,MSC09,MiadICNSC08}, \ab{vehicle traffic networks \cite{zeng2022multi,zeng2023congestion,qian2017traffic,zhang2023hybrid}}. This paper focuses on faster convergence towards optimal resource allocation.  Accelerated convergence is essential for rapidly adapting to changing conditions and ensuring timely responses to resource allocation requests. 
	Traditional centralized methods often prove inadequate in addressing the intricacies, efficiency, and scale of contemporary distributed large-scale systems. Therefore, the ongoing research has shifted towards decentralized algorithms seamlessly adapting to the ever-changing landscape of these systems \cite{yang2019survey,nedic2018distributed,assran2020advances}. Fast convergence is particularly critical where real-time or near-real-time performance is required, such as cloud computing, edge computing, Internet of Things (IoT) \cite{he2020large}, and in the context of distributed computing environments \cite{grammenos2023cpu}. Recently, fast distributed averaging algorithms are considered in the literature that further motivates this work \cite{xiao2004fast,zhang2020power}.
	
	The existing literature on distributed allocation and scheduling are classified into two main domains: \ab{linear primal-based  \cite{boyd2006optimal,cherukuri2015distributed,shames2011accelerated,cherukuri2016initialization} and dual-based solutions \cite{doan2017ccta,cdc_dtac,banjac2019decentralized,chang2016proximal,jian2019distributed,aybat2019distributed,ding2018primal,jiang2022distributed,shao2022distributed,xu2017distributed}, e.g.,  alternating-direction-method-of-multipliers (ADMM).  The primal-based solutions \cite{boyd2006optimal,shames2011accelerated,cherukuri2016initialization} are in general slower than ADMM-based setups; however, the primal-based solutions are all-time feasible \cite{boyd2006optimal,shames2011accelerated,cherukuri2016initialization}, i.e., the solution holds resource-demand feasibility at all times. The work \cite{shames2011accelerated} improves the convergence rate by adding momentum term as compared to \cite{boyd2006optimal,cherukuri2016initialization}. Dual-based solutions include: Lagrangian-based \cite{doan2017ccta}, dual consensus ADMM \cite{banjac2019decentralized,jian2019distributed,chang2016proximal}, asynchronous ADMM \cite{jiang2022distributed}, primal-dual ADMM-like \cite{aybat2019distributed}, primal-dual laplacian gradient flow \cite{ding2018primal}, proportional-integral control-based \cite{shao2022distributed}, and nonnegative surplus-based algorithm \cite{xu2017distributed}.} In the context of distributed machine-learning, some signum-based solutions \cite{shi2023finite,doan2019finite,zhang2018distributed} and distributed heavy-ball methods are proposed to improve the convergence rate \cite{xin2019distributed}. The all-time feasibility of primal-based solutions implies that at any termination time of the algorithm, the solution meets the resource-demand constraint feasibility, while in dual-based solutions, it takes some time to meet this constraint feasibility. Therefore, the dual-based solutions must be fast enough to satisfy the resource-demand constraint in time, otherwise, it may cause service disruption and breakdown. This is especially of interest in distributed generator coordination and energy resource management \cite{cherukuri2015distributed}. Further, ADMM algorithms suffer from all-time-connectivity requirements, i.e., they need to not lose connectivity at any time to avoid divergence and non-optimal solutions. 
	
	This work proposes a nonlinear Laplacian gradient solution for distributed allocation and scheduling via primal formulation. The proposed solution is all-time feasible, which means that at any iteration the resource-demand constraint holds. This avoids feasibility violation and service disruption. We prove convergence over uniformly-connected networks. This is important as in real \textit{mobile} multi-agent systems the agents may come into or leave the communication range of one another, which may cause dynamic networks with uniform-connectivity instead of all-time connectivity. The other reason is the possibility of packet drops and temporary link failure over the network. The main feature of our algorithm is its accelerated convergence, which is tunable by changing the associated parameters. The proposed nonlinear continuous-time dynamics improves the convergence rate of the existing primal-based solutions (over uniformly-connected networks) by two parameters associated with signum-based functions. As compared to most existing literature, this work (i) relaxes the all-time connectivity requirement in most literature \cite{boyd2006optimal,cherukuri2015distributed,shames2011accelerated,cherukuri2016initialization,doan2017ccta,cdc_dtac,banjac2019decentralized,jian2019distributed,aybat2019distributed,ding2018primal,jiang2022distributed,shao2022distributed,chang2016proximal,xu2017distributed} to uniform-connectivity, (ii) leads to an all-time feasible solution in contrast to asymptotic feasibility in dual-based solutions \cite{doan2017ccta,cdc_dtac,banjac2019decentralized,jian2019distributed,aybat2019distributed,ding2018primal,jiang2022distributed,shao2022distributed,chang2016proximal,xu2017distributed}, and (iii) improves the convergence rate of the existing all-time feasible primal-based solutions \cite{boyd2006optimal,cherukuri2015distributed,shames2011accelerated,cherukuri2016initialization}. More importantly, the points (i)-(iii) are addressed altogether and, to our best knowledge, no existing work in the literature addresses these. 


\section{Problem Formulation and Preliminaries} \label{sec_pre}
    \textbf{ Notation:} $\mb{1}_n$ denotes all ones vector of size $n$. $\nabla F({\mb{x}})$ denotes the gradient of function $F(\cdot)$ with respect to state parameter $\mb{x}$. $\partial \overline{F}(\mb{x})$ denotes the \textit{generalized derivative} of the nonsmooth function $\overline{F}(\mb{x})$. $\mc{L}_\mc{H} F : \mathbb{R}^n  \rightrightarrows  \mathbb{R}$ denotes the \textit{set-valued Lie-derivative}  of function $F$ with respect to the dynamics $\dot{\mb{x}} \in \partial \mc{H}(\mb{x})$. See detailed definitions in \cite{cortes2008discontinuous}.

\textbf{ Problem:} 
\ab{Consider a network of $n$ agents to be allocated with a share of overall resources $x_i$ equal to $b$. These resources are associated with a cost $f_i(x_i)$. The sum of overall resources is constant equal to $b$ and the allocation needs to optimize the overall cost. 
This problem is modelled as the following constrained optimization problem \cite{cherukuri2016initialization}:}
\begin{align} \label{eq_dra}
	\mc{P}_1:~~	\min_\mb{x}
	~~ & F_0(\mb{x}) = \sum_{i=1}^{n} f_i(x_i),~
	\text{s.t.} ~  \mb{1}_n^\top \mb{x}  - b = 0,
\end{align}
where $x_i$ is the state of the agent $i$ representing the amount of resources allocated to the agent $i$ and $\mb{x}:=[x_1;\dots;x_n]$.  The global state vector $\mb{x}$ denotes the column vector of all allocated states and $n$ is the number of agents.  The \textit{feasibility constraint} $\mb{1}_n^\top \mb{x}  - b = 0$ (or $\mb{1}_n^\top \mb{x}  = b$) implies that the sum of resources to be allocated is fixed and constrained with the value $b$. The functions $f_i(x_i)$ represent the cost of allocating resources to be minimized. These functions are strictly convex.
The agenr/node states might be locally constrained with some box constraints $
	m_i \leq x_i \leq M_i,
$
implying that the amount of allocated resource to node $i$ is upper-bounded by $M_i$ and lower-bounded by $m_i$. These local constraints can be added as penalizing convex terms to the cost function referred to as penalty functions or barrier functions \cite{wu2021new,bertsekas2003convex}. \ab{Then, considering the penalty term as $\widetilde{F}(\mb{x})$ the objective function changes into the form $F(\mb{x}) := F_0(\mb{x}) +\widetilde{F}(\mb{x})$. One example penalty function is \cite{slp_book,nesterov1998introductory}}
\begin{align}
	\sigma [{x}_i - M_i]^+ + \sigma [m_i - {x}_i]^+,
\end{align}
with $[u]^+ = \max \{u,0\}^c$, $c \in \mathbb{Z}^+$ \cite{bertsekas2003convex}. Note that this function is smooth for $c\geq2$, and for the non-smooth case of $c=1$ one can use the following smooth \ab{approximation \cite{slp_book,nesterov1998introductory}}:
\begin{align} \label{eq_log_penalty}
	L(u,\rho)=\frac{\sigma}{\rho} \log (1+\exp (\rho u)),
\end{align}
where $\sigma$ weights the box constraint penalty terms as compared to the main objective function $f_i(x_i)$. In other words, by setting $\sigma>1$ the solution is more toward satisfying the box constraint, while $\sigma<1$ puts less weight to satisfy the box constraint and thus more toward optimizing the cost/objective function.
It is proved that by choosing~$\rho$ large enough~${L(u,\rho)}$ gets arbitrarily close to~$\max\{u,0\}$ and the maximum gap between the two functions inversely scales with $\rho$ \cite{slp_book,nesterov1998introductory}. Similarly, other smooth and convex barrier functions can be found in \cite{wu2021new,bertsekas2003convex}. Note that the sum of strictly convex local cost functions and convex penalty terms is strictly convex. Further, the penalty functions are generally non-quadratic which implies that adding them to the objective function makes the optimization problem non-quadratic. This makes the problem different from and more challenging than the existing consensus-based quadratic problems, e.g., \cite{kar_edp}. It should be mentioned that using barrier or penalty terms is an approximate solution to address the local box constraints.

\begin{lem} \label{lem_equilibria}
	\cite{boyd2006optimal}
 	The strictly convex constrained optimization problem \eqref{eq_dra} has a unique optimal solution ${\mb{x}^*}$ for which the objective gradient satisfies
 	$ \nabla F({\mb{x}^*}) \in \mbox{span}(\mb{1}_n)$,
 	where $\nabla F({\mb{x}^*}) := (\frac{df_1}{dx_1}(x_1^*),\dots,\frac{df_n}{dx_n}(x_n^*))^\top$.
 \end{lem}
\begin{proof}
	The proof follows the so-called Karush-Kuhn-Tucker (KKT) conditions, \ab{i.e., the optimizer is a factor of $\mb{1}_n$ as the gradient of the linear constraint $\mb{1}_n^\top \mb{x}  - b = 0$ \cite{bertsekas2003convex}.}
\end{proof}

\section{The Proposed Accelerated Algorithm} \label{sec_alg}
For the distributed setup we consider an undirected network $\mc{G}$ of $n$ agents communicating optimization data over weighted links. The adjacency matrix of this network is denoted by $W$ and is symmetric. We assume the network is uniformly connected, i.e., its union over finite time intervals is strongly connected while it is not necessarily connected at every time instant.
Then, our proposed accelerated continuous-time dynamics is as follows:

\small
\begin{align}
	\dot{{x}}_i =  -  \eta \sum_{j \in \mc{N}_i} W_{ij}\Bigl( \mbox{sgn}^{\alpha}\Bigl(\frac{d f_i}{d {x}_i}-\frac{d f_j}{d {x}_j}\Bigr) 
	+  \mbox{sgn}^{\beta}\Bigl(\frac{d f_i}{d {x}_i}-\frac{d f_j}{d {x}_j}\Bigr)\Bigr)
	\label{eq_solution0}
\end{align}\normalsize
with $0<\alpha<1$ and $\beta>1$ as rate tuning parameters, $W_{ij}=W_{ji}\geq 0$ as symmetric link weight factors, $\eta$ as the positive step-rate, $\mc{N}_i$ as the set of neighbors of $i$ over the multi-agent network $\mc{G}$, and signum (or sign) function defined as \cite{shi2023finite} 
\begin{align}
	\mbox{sgn}^{\alpha}(u)=u|u|^{\alpha-1}
\end{align}
and $\mbox{sgn}^{\beta}(\cdot)$ similarly follows. Based on the definition of $\mbox{sgn}^{\alpha}(\cdot)$ and $\mbox{sgn}^{\beta}(\cdot)$ one can rewrite the dynamics~\eqref{eq_solution0} as,
\begin{align} \nonumber
	\dot{{x}}_i = \left\{ \begin{array}{ll}
		-\eta \sum_{j \in \mc{N}_i} W_{ij}\Bigl(\mbox{sgn}^{\alpha}(u) 
		+ \mbox{sgn}^{\beta}(u) \Bigr), & \text{if}~  u\geq 0,\\
		\eta \sum_{j \in \mc{N}_i} W_{ij}\Bigl( \mbox{sgn}^{\alpha}(u) 
		+  \mbox{sgn}^{\beta}(u) \Bigr), & \text{if}~  u< 0.
	\end{array}\right.
\end{align}
with $u=\frac{d f_i}{d {x}_i}-\frac{d f_j}{d {x}_j}$. The proposed gradient tracking equation in discrete time is in the following form:

\small
\begin{align} \nonumber
	{x}_i(k+1) = {x}_i(k) - \eta \sum_{j \in \mc{N}_i} &W_{ij}\Bigl( \mbox{sgn}^{\alpha}(\frac{d f_i(x_i(k))}{d x_i}-\frac{d f_j(x_j(k))}{d x_j}) 
	\\ 
	&+  \mbox{sgn}^{\beta}(\frac{d f_i(x_i(k))}{d x_i}-\frac{d f_j(x_j(k))}{d x_j})\Bigr)
	\label{eq_sol_d}
\end{align}\normalsize
with $k$ as the discrete-time index. \ab{For gradient descent in \textit{discrete-time}, it is known that the step size $\eta$ should satisfy $0 < \eta < 1/L$ with $L$ as the Lipschitz constant of $F(\cdot)$. } Note that, following Lemma~\ref{lem_equilibria}, to reach the optimizer $\mb{x}^*$ the gradient $\frac{d f_i}{d {x}_i}$ at all agents $i$ need to reach consensus. Therefore, the proposed dynamics drive the states at all agents such that the difference of their gradients converges to zero, and at the optimal point we have $ \nabla F({\mb{x}^*}) \in \mbox{span}(\mb{1}_n)$.
\begin{rem}
	Signum functions $\mbox{sgn}^{\alpha}(\cdot)$ and $\mbox{sgn}^{\beta}(\cdot)$ accelerate the convergence rate of the dynamics \eqref{eq_solution0} and \eqref{eq_sol_d} toward the optimal point. The parameters $\alpha$ and $\beta$ tune the rate of convergence. By decreasing parameter $\alpha$ and increasing parameter $\beta$ both dynamics converge faster. \ab{However, this may cause higher residual and chattering as discussed later. So, there is a trade-off between the convergence rate and steady-state residual.} In fact, the term $\mbox{sgn}^{\beta}(\cdot)$ improves the convergence rate in the regions far from the optimizer $\mb{x}^*$ and, on the other hand, $\mbox{sgn}^{\alpha}(\cdot)$ improves the convergence rate in the close vicinity of the optimizer $\mb{x}^*$. This is because we have $\mbox{sgn}^{\alpha}(\frac{d f_i}{d {x}_i}-\frac{d f_j}{d {x}_j})>\frac{d f_i}{d {x}_i}-\frac{d f_j}{d {x}_j}$ for $|\frac{d f_i}{d {x}_i}-\frac{d f_j}{d {x}_j}|<1$ and this implies that in these regions (closer to the optimal equilibrium) the convergence is faster. Similarly, we have $\mbox{sgn}^{\beta}(\frac{d f_i}{d {x}_i}-\frac{d f_j}{d {x}_j})>\frac{d f_i}{d {x}_i}-\frac{d f_j}{d {x}_j}$ for $|\frac{d f_i}{d {x}_i}-\frac{d f_j}{d {x}_j}|>1$ and this makes the convergence faster in regions far from the equilibrium in Lemma~\ref{lem_equilibria}.  Therefore, the convergence rate is faster than the linear case; however, it is not uniform in all regions due to the nonlinearity (this is also shown by simulation). In the linear case \cite{cherukuri2015distributed} with $\alpha=1$ and $\beta=1$, the convergence rate is $\mc{O}(\frac{1}{k})$.
\end{rem}

In the rest of the paper, we prove feasibility and convergence for continuous-time dynamics~\eqref{eq_solution0} and equivalent results hold for the discrete-time dynamics~\eqref{eq_sol_d}.  
The following lemma proves the all-time feasibility of the proposed solution.
 
\begin{lem} \label{lem_feasible}
	Initializing with a feasible state satisfying the constraint $\mb{1}_n^\top \mb{x}(0)  - b = 0$, the dynamics~\eqref{eq_sol_d} remains all-time feasible for $t>0$.
\end{lem}
\begin{proof}
	For any $t>0$ under dynamics~\eqref{eq_sol_d} the feasibility constraint satisfies the following,
\begin{align} \nonumber
	\mb{1}_n^\top \dot{\mb{x}} =  \sum_{i =1}^n \dot{{x}}_i =& - \eta \sum_{i =1}^n \sum_{j \in \mc{N}_i} W_{ij}\Bigl( \mbox{sgn}^{\alpha}(\frac{d f_i}{d x_i}-\frac{d f_j}{d x_j}) 
	\\ \label{eq_feas_prov}
	&+  \mbox{sgn}^{\beta}(\frac{d f_i}{d x_i}-\frac{d f_j}{d x_j})\Bigr).
\end{align}
	Note that the signum function is odd and sign-preserving. Therefore, $\mbox{sgn}^{\alpha}(-u) = - \mbox{sgn}^{\alpha}(u)$ and $\mbox{sgn}^{\beta}(-u) = - \mbox{sgn}^{\beta}(u)$ with $u=\frac{d f_i}{d {x}_i}-\frac{d f_j}{d {x}_j}$. Also, the link weights are symmetric and we have $W_{ij}=W_{ji}$. Therefore, $W_{ij}\mbox{sgn}^{\alpha}(u)+W_{ji}\mbox{sgn}^{\alpha}(-u)=0$ and the summation in the right-hand-side of Eq.~\eqref{eq_feas_prov} is zero. Then, we have
	$\mb{1}_n^\top\dot{\mb{x}} = 0$. Since the solution is initially feasible, this proves all-time feasibility.	
\end{proof}


\begin{lem} \label{lem_tree}
	If the multi-agent network $\mc{G}(t)$ is uniformly connected, the optimal solution $\mb{x}^*$ given by Lemma~\ref{lem_equilibria} is the unique equilibrium of \eqref{eq_sol_d}.
\end{lem}

\begin{proof}
	We prove this by contradiction. Assume another equilibrium $\widehat{\mb{x}}$ for dynamics \eqref{eq_sol_d} for which $ \nabla F({\widehat{\mb{x}}}) \notin \mbox{span}(\mb{1}_n)$. This implies that for this point $\frac{d f_i}{d \mb{x}_i}\neq\frac{d f_j}{d \mb{x}_j}$ for at least one pair of agents $i,j$. Assume $\nabla F({\widehat{\mb{x}}}) = (\widehat{\phi}_1,\dots,\widehat{\phi}_n)^\top$ and find the agents $a = \mbox{argmax}_{\lambda\in \{1,\dots,n\}}  \widehat{\phi}_{\lambda }$ and $b = \mbox{argmin}_{\lambda \in \{1,\dots,n\}}  \widehat{\phi}_\lambda$. From the contradiction we have $\widehat{\phi}_{a}> \widehat{\phi}_{b}$. Uniform connectivity of the network $\mc{G}(t)$ implies that there exists a finite time-interval $T$ for which there exists a path over $\cup_{t}^{t+T} \mc{G}(t)$ from $a$ to $b$. This path includes (at least) two agents $a',b'$ such that  $\widehat{\phi}_{a'}> \widehat{\phi}_{b'}$. Then, in a subdomain of $[t,t+T]$, \ab{there exists at least a neighbouring agent $l$ for which $\widehat{\phi}_{a'}> \widehat{\phi}_{l}$ and from the dynamics \eqref{eq_sol_d} we have  
 $\sum_{l \in \mc{N}_{a'}} \mbox{sgn}^{\alpha}(\widehat{\phi}_{a'} -\widehat{\phi}_{l}) + \mbox{sgn}^{\beta}(\widehat{\phi}_{a'} -\widehat{\phi}_{l})>0$; similarly, there is a neighbouring agent $l$ for which $\widehat{\phi}_{b'}< \widehat{\phi}_{l}$ and  $\sum_{l \in \mc{N}_{b'}} \mbox{sgn}^{\alpha}(\widehat{\phi}_{b'}-\widehat{\phi}_{l})+ \mbox{sgn}^{\beta}(\widehat{\phi}_{b'}-\widehat{\phi}_{l})<0$.} Therefore, $\dot{\widehat{\mb{x}}}_{a'} < 0$ and   $\dot{\widehat{\mb{x}}}_{b'} > 0$. This contradicts the equilibrium assumption for~$\widehat{\mb{x}}$ and proves the lemma.
\end{proof}
Note that, as in consensus algorithms, uniform connectivity is the least requirement for convergence toward optimal value.
\begin{lem} \label{lem_sum}
	Let $W$ be a symmetric adjacency matrix of $\mc{G}$. Then, for $\phi \in \mathbb{R}^n$ the following holds,
	\begin{align} \nonumber
		\sum_{i =1}^n \phi_i\sum_{j =1}^n W_{ij} \mbox{sgn}^{\alpha}(\phi_i &-\phi_j) = \\
		- \frac{1}{2} \sum_{i,j=1}^n W_{ij} &|\phi_i -\phi_j|^{\alpha+1}.
	\end{align}
Similar equation holds for $\mbox{sgn}^{\beta}(\cdot)$.
\end{lem}
\begin{proof}
	We have $W_{ij} = W_{ji}$, and $\mbox{sgn}^{\alpha}(\cdot)$, $\mbox{sgn}^{\beta}(\cdot)$ as sign-preserving odd functions. Therefore,
	\begin{align} \nonumber
		\phi_i W_{ij} & \mbox{sgn}^{\alpha}(\phi_j -\phi_i) + \phi_j W_{ji}  \mbox{sgn}^{\alpha}(\phi_i-\phi_j)\\ \nonumber
		& = W_{ij}(\phi_i-\phi_j) \mbox{sgn}^{\alpha}(\phi_j-\phi_i) \\
		& =- W_{ij} |\phi_i-\phi_j|^{\alpha+1}.
	\end{align}
	The proof similarly follows for $\mbox{sgn}^{\beta}(\cdot)$. 
\end{proof}


\begin{theorem} \label{thm_converg}
	For feasible initialization over uniformly-connected network $\mc{G}$, dynamics \eqref{eq_solution0} converges to the optimizer of allocation/scheduling problem \eqref{eq_dra}.
\end{theorem}

\begin{proof}	
	The proof is based on nonsmooth Lyapunov analysis in \cite{cortes2008discontinuous}.  Denote by $F^*=F(x^*)$ the optimal value of the cost function in \eqref{eq_dra}. Define the Lyapunov function as the residual $\overline{F}(\mb{x}(t))=F(\mb{x}(t))-F^*$ at every time $t$ along the solution of dynamics \eqref{eq_solution0}. For this positive Lyapunov function, the unique equilibrium is $\overline{F}(\mb{x}^*)=0$. Following  \cite[Proposition~10]{cortes2008discontinuous},
	the derivative of this Lyapunov function $\overline{F}(\mb{x})$ \ab{satisfies $\partial \overline{F}(\mb{x}(t)) \in \mc{L}_\mc{H} \overline{F}(\mb{x}(t)) \in \mathbb{R}$} \cite[Proposition~10]{cortes2008discontinuous} where $\mc{H}$ refers to the solution dynamics \eqref{eq_solution0}. Then, the generalized gradient of the residual is as follows:	
	\begin{align}\nonumber
		\partial \overline{F}(\mb{x}) = &\nabla \overline{F}(\mb{x}) \dot{\mb{x}} = \sum_{i =1}^n \frac{df_i}{dx_i}\dot{\mb{x}}_i \\
		=  &\sum_{i =1}^n \frac{df_i}{dx_i}\Bigl(- \eta \sum_{j \in \mc{N}_i} W_{ij}\Bigl(\mbox{sgn}^{\alpha}(\frac{df_i}{dx_i}-\frac{df_j}{dx_j}) \nonumber \\ &+  \mbox{sgn}^{\beta}(\frac{df_i}{dx_i}-\frac{df_j}{dx_j})\Bigr)\Bigr). \nonumber
	\end{align}
	Then, using Lemma~\ref{lem_sum},
	\begin{align} \label{eq_Frate}
		\partial \overline{F}(\mb{x}) =  &-\frac{\eta}{2}\Bigl(\sum_{i,j =1}^n  W_{ij}|\frac{df_i}{dx_i}-\frac{df_j}{dx_j}|^{\alpha+1} \nonumber \\ &+ \sum_{i,j =1}^n  W_{ij} |\frac{df_i}{dx_i}-\frac{df_j}{dx_j}|^{\beta+1}\Bigr).
	\end{align}
The above implies that $\partial \overline{F}(\mb{x}) \leq 0$ and the Lyapunov function is non-increasing under dynamics \eqref{eq_solution0}. The invariant accumulation set under the proposed dynamics includes the state values $\partial \overline{F}(\mb{x})  = 0$, i.e.,  $\{\mb{x}^*|\nabla F(\mb{x}^*) \in \mbox{span}\{\mb{1}_n\}\}$. This follows from Lemma~\ref{lem_equilibria}. Following the all-time feasibility of the solution from Lemma~\ref{lem_feasible} the equilibrium also satisfies $\mb{1}_n^\top \mb{x}^*  - b = 0$. This implies that the unique optimizer satisfying $\mb{1}_n^\top \mb{x}^* = b$ and $\nabla F(\mb{x}^*) \in \mbox{span}\{\mb{1}_n\}$ is the equilibrium of dynamics \eqref{eq_solution0} ($\partial \overline{F}(\mb{x}) = 0$) and for other points $\partial \overline{F}(\mb{x}) < 0$. \ab{This $\partial \overline{F}(\mb{x}) < 0$ holds if there exist few links over the network and does not require all-time connectivity.}
 Thus, using LaSalle's invariance principle \cite[Theorem~2.1]{cherukuri2015distributed}, the solution converges to the set $\mc{I}$ contained in $\{\mb{x} | 0 \in \mc{L}_\mc{H} \overline{F}(\mb{x}(t)),\mb{1}_n^\top \mb{x} = b \}$ with $\mc{L}_\mc{H} \overline{F}(\mb{x}(t))$ denoting the Lie derivative of the residual with respect to Eq. \eqref{eq_solution0}. Since $\mc{I} = \{\mb{x}^*\}$, and $ \max \mc{L}_\mc{H} \overline{F}(\mb{x}(t)) < 0$ for $\mb{x} \notin \mc{I}$, dynamics \eqref{eq_solution0} globally asymptotically converges to $\mc{I} = \{\mb{x}^*\}$  \cite[Theorem~1]{cortes2008discontinuous}. This completes the proof.
\end{proof}
Finally, we summarize our solution in the Algorithm~\ref{alg_1}.

\begin{algorithm} 
	\textbf{Input:}  $\mc{N}_i$, $W$, $\eta$, $m_i$, $M_i$, $b$, $f_i(\cdot)$\;
	\textbf{Initialization:} $t=0$, random feasible $x_i(0)$\;
	\While{algorithm running}{
		Agent $i$ receives the gradient $\frac{df_j}{dx_j}$ from agents in $j \in \mc{N}_i$ over the multi-agent network $\mc{G}$\;
		Agent $i$ computes Eq.~\eqref{eq_solution0} (or Eq.~\eqref{eq_sol_d})\;
		Agent $i$ sends its gradient $\frac{df_i}{dx_i}$ to neighboring agents $j$ where $i \in \mc{N}_j$ over $\mc{G}$\;
	}
	\textbf{Return} Final state $x^*_i$ and objective function $f_i(x^*_i)$\;	
	\caption{The Accelerated Allocation Algorithm} 
	\label{alg_1}
\end{algorithm}

It should be noted that the discrete-time dynamics \eqref{eq_sol_d} may result in steady-state residual due to the non-Lipschitz nature of the sign-based function. This residual is closely dependent on step-rate $\eta$ and parameters $\alpha$ and $\beta$. \ab{To be more specific, larger $\eta$, smaller (close to zero) $\alpha$, and larger $\beta$ result in larger steady-state residuals. So there is a trade-off between the optimality gap and convergence rate of the proposed dynamics in discrete-time, i.e., the faster convergent solution may result in a higher residual.} This is better shown by the simulations in the next section. Note that this is not an issue in the continuous-time case.

\section{Simulations} \label{sec_sim}
For simulation, we choose a random \ab{time-varying Erdos-Renyi network of $n=50$ nodes with $p=20\%$ linking probability}. The cost function at agent $i$ is $f_i(x_i) = a_i x_i^2 + b_i x_i$ with randomly chosen $a_i \in (0,0.3]$ and $b_i \in (0,10]$. The box constraints are $m_i=20, M_i=105$ addressed in the objective function via logarithmic penalty term \eqref{eq_log_penalty} with $\rho=1,\sigma=1$. By setting $\sigma=1$ we equally weight the objective function as compared to the penalty term (for box constraint). Also, note that the algorithm works for any $\rho$ and we choose $\rho=1$ as an example here. The states are initialized with random values satisfying the feasibility condition $\mb{1}_n^\top \mb{x}(0) = b = 3000$.
In Fig.~\ref{fig_sim}, the residual under dynamics~\eqref{eq_solution0}  is compared with some primal-based all-time feasible solutions proposed in the literature, namely, linear \cite{cherukuri2015distributed}, accelerated linear \cite{shames2011accelerated}, finite-time \cite{chen2016distributed}, and saturated \cite{ccta,scl} solutions. Note that these primal-based solutions are all-time feasible in contrast to the recent dual-based (e.g., ADMM) solutions \cite{doan2017ccta,cdc_dtac,banjac2019decentralized,jian2019distributed,aybat2019distributed,ding2018primal,jiang2022distributed,shao2022distributed,chang2016proximal,xu2017distributed}. This is the reason behind comparing our work with the mentioned primal-based literature.
For our dynamics we set $\alpha=0.3$, $\beta=1.7$, and $\eta=0.2$. Recall that dual-based solutions are not all-time feasible. As it is clear from the figure the decay rate of the proposed solution is faster than the mentioned works.
The time evolution of the states is also shown in the figure. 
As shown in the figure, the average of states remains unchanged as states evolve over time, implying that the sum of states is constant and the solution is all-time feasible. This implies that at every time-instant the resource-demand feasibility constraint $\mb{1}_n^\top \mb{x}(t) = b = 3000$ holds.

\begin{figure}[]
	\center 
	\includegraphics[width=1.72in]{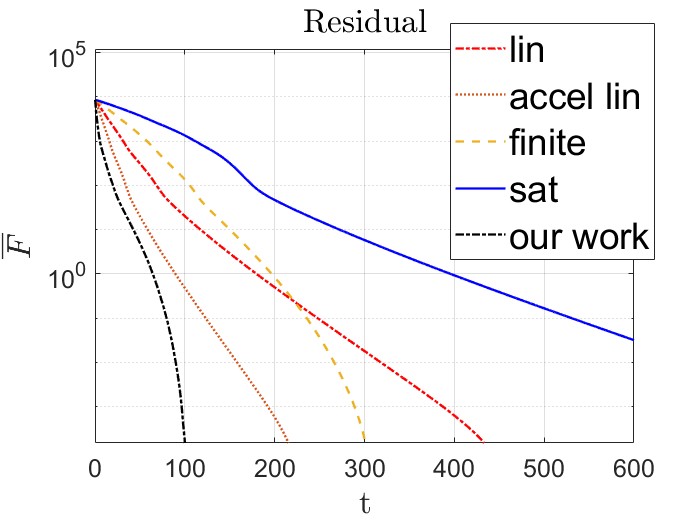}
	\includegraphics[width=1.72in]{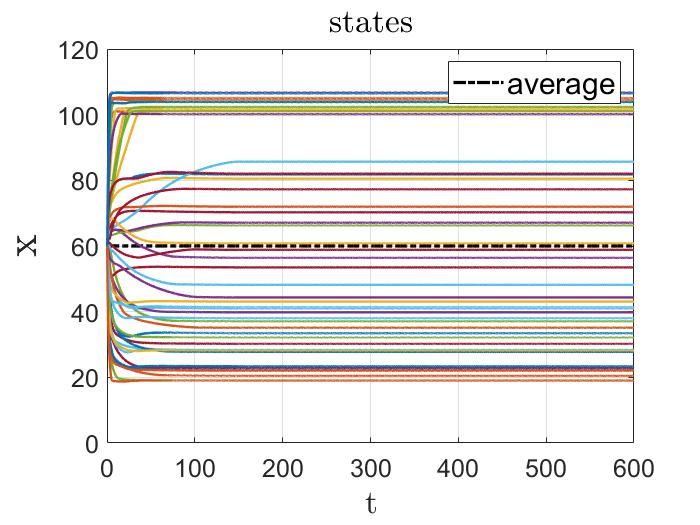}
	\caption{\ab{(Left) Time-evolution of the residual value under the proposed accelerated dynamics~\eqref{eq_solution0} as compared with some existing literature: linear \cite{cherukuri2015distributed}, accelerated linear with $b=0.5$ \cite{shames2011accelerated}, finite-time with $\nu=0.7$ \cite{chen2016distributed}, and saturated with $\delta =1$ \cite{ccta,scl} solutions  (Right) Time-evolution of the state values of all agents under the proposed dynamics. }}
	\label{fig_sim} 
\end{figure}

Next, we simulate the residual evolution for different $\alpha$ and $\beta$ values and fixed $\eta=0.1$ to tune the convergence rate. The cost and network parameters are set the same as in the previous simulation. The network $\mc{G}$ is switched between $6$ Erdos-Renyi networks every $1$ sec while the union of the $6$ networks is connected (implying uniform connectivity). The time evolution of the residuals and states are shown in Fig.~\ref{fig_sim2}-(Left). It is clear that increasing $\beta$ and decreasing $\alpha$ improves the convergence rate. For $\beta=1$ and $\alpha=1$ the algorithm represents the linear case provided for comparison. Finally, Fig.~\ref{fig_sim2}-(Right) represents the evolution of the discrete-time dynamics~\eqref{eq_sol_d} for different values of $\alpha$, $\beta$, and $\eta$. Larger $\eta$, smaller (close to zero) $\alpha$, and larger $\beta$ result in faster convergence despite larger steady-state residuals. This shows the trade-off between the convergence rate and the optimality gap.

\begin{figure}[]
	\center 
	\includegraphics[width=1.72in]{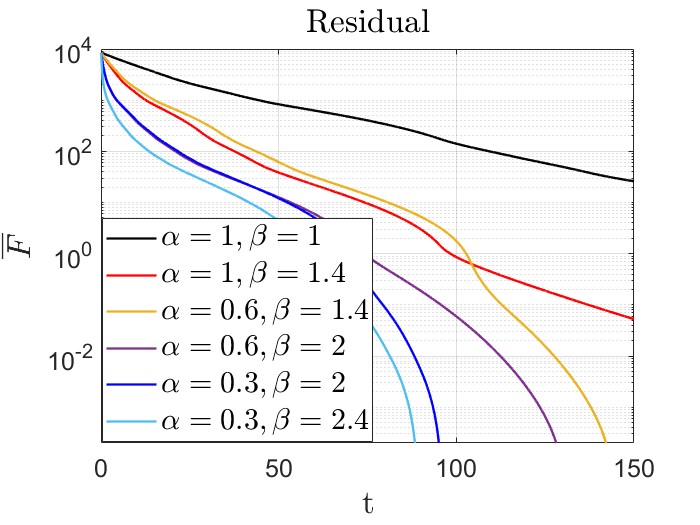}
	\includegraphics[width=1.72in]{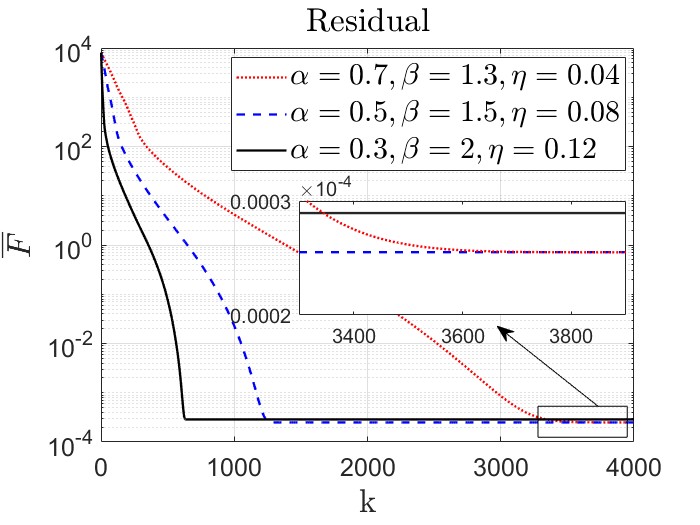}
	\caption{\ab{(Left) The residual decay rate for different $\alpha$ and $\beta$ values under dynamics~\eqref{eq_solution0}. (Right) The residual evolution under discrete-time dynamics~\eqref{eq_sol_d} for different $\alpha$, $\beta$, and $\eta$ values.}}
	\label{fig_sim2} 
\end{figure}

\section{Conclusions} \label{sec_con}
This paper presents fast all-time feasible allocation and scheduling algorithm over uniformly-connected networks, advancing the dual-based solutions in terms of connectivity and feasibility. Also, it advances the primal-based solution in terms of uniform-connectivity and accelerated convergence. 
As a future direction, one can consider other types of nonlinear solutions. For example, it is known that signum-based solutions are robust to noise/disturbances \cite{paganini_book}. The uniform connectivity also allows for convergence over unreliable networks with packet drops. Resource allocation in the presence of malicious agents \cite{wang2022byzantine} is another future direction.
	
\bibliographystyle{IEEEbib}
\bibliography{bibliography}

\begin{thebibliography}{10}

\bibitem{ojsys}
M.~Doostmohammadian, A.~Aghasi, A.~I. Rikos, A.~Grammenos, E.~Kalyvianaki,
  C.~N. Hadjicostis, K.~H. Johansson, and T.~Charalambous,
\newblock ``Distributed anytime-feasible resource allocation subject to
  heterogeneous time-varying delays,''
\newblock {\em IEEE Open Journal of Control Systems}, vol. 1, pp. 255--267,
  2022.

\bibitem{ccta2}
M.~Doostmohammadian, A.~Aghasi, A.~Rikos, A.~Grammenos, E.~Kalyvianaki,
  C.~Hadjicostis, K.~Johansson, and T.~Charalambous,
\newblock ``Distributed cpu scheduling subject to nonlinear constraints,''
\newblock in {\em IEEE Conf. on Control Technology and Applications}, 2022, pp.
  746--751.

\bibitem{grammenos2023cpu}
A.~Grammenos, T.~Charalambous, and E.~Kalyvianaki,
\newblock ``{CPU} scheduling in data centers using asynchronous finite-time
  distributed coordination mechanisms,''
\newblock {\em IEEE Trans. on Network Science and Engineering}, 2023.

\bibitem{cherukuri2015distributed}
A.~Cherukuri and J.~Cort{\'e}s,
\newblock ``Distributed generator coordination for initialization and anytime
  optimization in economic dispatch,''
\newblock {\em IEEE Transactions on Control of Network Systems}, vol. 2, no. 3,
  pp. 226--237, 2015.

\bibitem{derm}
M.~Doostmohammadian,
\newblock ``Distributed energy resource management: All-time resource-demand
  feasibility, delay-tolerance, nonlinearity, and beyond,''
\newblock {\em IEEE Control Systems Letters}, 2023.

\bibitem{kar_edp}
S.~Kar, G.~Hug, J.~Mohammadi, and J.~M.~F. Moura,
\newblock ``Distributed state estimation and energy management in smart grids:
  A consensus ${+}$ innovations approach,''
\newblock {\em IEEE Journal of Selected Topics in Signal Processing}, vol. 8,
  no. 6, pp. 1022--1038, 2014.

\bibitem{alizadeh2012demand}
M.~Alizadeh, X.~Li, Z.~Wang, A.~Scaglione, and R.~Melton,
\newblock ``Demand-side management in the smart grid: Information processing
  for the power switch,''
\newblock {\em IEEE Signal Processing Magazine}, vol. 29, no. 5, pp. 55--67,
  2012.

\bibitem{zhang2019distributed}
J.~Zhang, A.~I. Koutrouvelis, R.~Heusdens, and R.~C. Hendriks,
\newblock ``Distributed rate-constrained {LCMV} beamforming,''
\newblock {\em IEEE Signal Processing Letters}, vol. 26, no. 5, pp. 675--679,
  2019.

\bibitem{MiadCons}
H.~Sayyaadi and M.~Moarref,
\newblock ``A distributed algorithm for proportional task allocation in
  networks of mobile agents,''
\newblock {\em IEEE Transactions on Automatic Control}, vol. 56, no. 2, pp.
  405--410, Feb. 2011.

\bibitem{MSC09}
M.~Doostmohammadian, H.~Sayyaadi, and M.~Moarref,
\newblock ``A novel consensus protocol using facility location algorithms,''
\newblock in {\em IEEE Conf. on Control Applications \& Intelligent Control},
  2009, pp. 914--919.

\bibitem{MiadICNSC08}
M.~Moarref and H.~Sayyaadi,
\newblock ``Facility location optimization via multi-agent robotic systems,''
\newblock in {\em IEEE International Conference on Networking, Sensing and
  Control}. IEEE, 2008, pp. 287--292.

\bibitem{zeng2022multi}
J.~Zeng, Y.~Qian, F.~Yin, L.~Zhu, and D.~Xu,
\newblock ``A multi-value cellular automata model for multi-lane traffic flow
  under lagrange coordinate,''
\newblock {\em Computational and Mathematical Organization Theory}, pp. 1--15,
  2022.

\bibitem{zeng2023congestion}
J.~Zeng, Y.~Qian, J.~Li, Y.~Zhang, and D.~Xu,
\newblock ``Congestion and energy consumption of heterogeneous traffic flow
  mixed with intelligent connected vehicles and platoons,''
\newblock {\em Physica A: Statistical Mechanics and its Applications}, vol.
  609, pp. 128331, 2023.

\bibitem{qian2017traffic}
Y.~Qian, J.~Zeng, N.~Wang, J.~Zhang, and B.~Wang,
\newblock ``A traffic flow model considering influence of car-following and its
  echo characteristics,''
\newblock {\em Nonlinear Dynamics}, vol. 89, pp. 1099--1109, 2017.

\bibitem{zhang2023hybrid}
J.~Zhang, Y.~Qian, J.~Zeng, X.~Wei, and H.~Li,
\newblock ``Hybrid characteristics of heterogeneous traffic flow mixed with
  electric vehicles considering the amplitude of acceleration and
  deceleration,''
\newblock {\em Physica A: Statistical Mechanics and its Applications}, vol.
  614, pp. 128556, 2023.

\bibitem{yang2019survey}
T.~Yang, X.~Yi, J.~Wu, Y.~Yuan, D.~Wu, Z.~Meng, Y.~Hong, H.~Wang, Z.~Lin, and
  K.~H. Johansson,
\newblock ``A survey of distributed optimization,''
\newblock {\em Annual Reviews in Control}, vol. 47, pp. 278--305, 2019.

\bibitem{nedic2018distributed}
A.~Nedi{\'c} and J.~Liu,
\newblock ``Distributed optimization for control,''
\newblock {\em Annual Review of Control, Robotics, and Autonomous Systems},
  vol. 1, pp. 77--103, 2018.

\bibitem{assran2020advances}
M.~Assran, A.~Aytekin, H.~Feyzmahdavian, M.~Johansson, and M.~G. Rabbat,
\newblock ``Advances in asynchronous parallel and distributed optimization,''
\newblock {\em Proceedings of the IEEE}, vol. 108, no. 11, pp. 2013--2031,
  2020.

\bibitem{he2020large}
Y.~He, S.~Zhang, L.~Tang, and Y.~Ren,
\newblock ``Large scale resource allocation for the {Internet of Things}
  network based on {ADMM},''
\newblock {\em IEEE Access}, vol. 8, pp. 57192--57203, 2020.

\bibitem{xiao2004fast}
L.~Xiao and S.~Boyd,
\newblock ``Fast linear iterations for distributed averaging,''
\newblock {\em Systems \& Control Letters}, vol. 53, no. 1, pp. 65--78, 2004.

\bibitem{zhang2020power}
J.~Zhang,
\newblock ``Power optimized and power constrained randomized gossip approaches
  for wireless sensor networks,''
\newblock {\em IEEE Wireless Communications Letters}, vol. 10, no. 2, pp.
  241--245, 2020.

\bibitem{boyd2006optimal}
L.~Xiao and S.~Boyd,
\newblock ``Optimal scaling of a gradient method for distributed resource
  allocation,''
\newblock {\em Journal of Optimization Theory and Applications}, vol. 129, no.
  3, pp. 469--488, 2006.

\bibitem{shames2011accelerated}
E.~Ghadimi, M.~Johansson, and I.~Shames,
\newblock ``Accelerated gradient methods for networked optimization,''
\newblock in {\em American Control Conference}. IEEE, 2011, pp. 1668--1673.

\bibitem{cherukuri2016initialization}
A.~Cherukuri and J.~Cort{\'e}s,
\newblock ``Initialization-free distributed coordination for economic dispatch
  under varying loads and generator commitment,''
\newblock {\em Automatica}, vol. 74, pp. 183--193, 2016.

\bibitem{doan2017ccta}
T.~T. Doan and C.~L. Beck,
\newblock ``Distributed lagrangian methods for network resource allocation,''
\newblock in {\em IEEE Conference on Control Technology and Applications
  (CCTA)}. IEEE, 2017, pp. 650--655.

\bibitem{cdc_dtac}
M.~Doostmohammadian, W.~Jiang, and T.~Charalambous,
\newblock ``{DTAC-ADMM}: Delay-tolerant augmented consensus {ADMM}-based
  algorithm for distributed resource allocation,''
\newblock in {\em IEEE 61st Conference on Decision and Control (CDC)}. IEEE,
  2022, pp. 308--315.

\bibitem{banjac2019decentralized}
G.~Banjac, F.~Rey, P.~Goulart, and J.~Lygeros,
\newblock ``Decentralized resource allocation via dual consensus {ADMM},''
\newblock in {\em American Control Conference (ACC)}. IEEE, 2019, pp.
  2789--2794.

\bibitem{chang2016proximal}
T.~Chang,
\newblock ``A proximal dual consensus {ADMM} method for multi-agent constrained
  optimization,''
\newblock {\em IEEE Transactions on Signal Processing}, vol. 64, no. 14, pp.
  3719--3734, 2016.

\bibitem{jian2019distributed}
L.~Jian, J.~Hu, J.~Wang, and K.~Shi,
\newblock ``Distributed inexact dual consensus {ADMM} for network resource
  allocation,''
\newblock {\em Optimal Control Applications and Methods}, vol. 40, no. 6, pp.
  1071--1087, 2019.

\bibitem{aybat2019distributed}
N.~S. Aybat and E.~Yazdandoost~Hamedani,
\newblock ``A distributed {ADMM}-like method for resource sharing over
  time-varying networks,''
\newblock {\em SIAM Journal on Optimization}, vol. 29, no. 4, pp. 3036--3068,
  2019.

\bibitem{ding2018primal}
D.~Ding and M.~R. Jovanovi{\'c},
\newblock ``A primal-dual laplacian gradient flow dynamics for distributed
  resource allocation problems,''
\newblock in {\em Annual American Control Conference (ACC)}. IEEE, 2018, pp.
  5316--5320.

\bibitem{jiang2022distributed}
W.~Jiang, M.~Doostmohammadian, and T.~Charalambous,
\newblock ``Distributed resource allocation via {ADMM} over digraphs,''
\newblock in {\em IEEE 61st Conference on Decision and Control}. IEEE, 2022,
  pp. 5645--5651.

\bibitem{shao2022distributed}
B.~Shao, M.~Li, and X.~Shi,
\newblock ``Distributed resource allocation algorithm for general linear
  multiagent systems,''
\newblock {\em IEEE Access}, vol. 10, pp. 74691--74701, 2022.

\bibitem{xu2017distributed}
Y.~Xu, T.~Han, K.~Cai, Z.~Lin, G.~Yan, and M.~Fu,
\newblock ``A distributed algorithm for resource allocation over dynamic
  digraphs,''
\newblock {\em IEEE Transactions on Signal Processing}, vol. 65, no. 10, pp.
  2600--2612, 2017.

\bibitem{shi2023finite}
X.~Shi, G.~Wen, and X.~Yu,
\newblock ``Finite-time convergent algorithms for time-varying distributed
  optimization,''
\newblock {\em IEEE Control Systems Letters}, 2023.

\bibitem{doan2019finite}
T.~Doan, S.~Maguluri, and J.~Romberg,
\newblock ``Finite-time analysis of distributed {TD(0)} with linear function
  approximation on multi-agent reinforcement learning,''
\newblock in {\em International Conference on Machine Learning}. PMLR, 2019,
  pp. 1626--1635.

\bibitem{zhang2018distributed}
J.~Zhang, K.~You, and T.~Ba{\c{s}}ar,
\newblock ``Distributed discrete-time optimization in multiagent networks using
  only sign of relative state,''
\newblock {\em IEEE Transactions on Automatic Control}, vol. 64, no. 6, pp.
  2352--2367, 2018.

\bibitem{xin2019distributed}
R.~Xin and U.~A. Khan,
\newblock ``Distributed heavy-ball: A generalization and acceleration of
  first-order methods with gradient tracking,''
\newblock {\em IEEE Transactions on Automatic Control}, vol. 65, no. 6, pp.
  2627--2633, 2019.

\bibitem{cortes2008discontinuous}
J.~Cortes,
\newblock ``Discontinuous dynamical systems,''
\newblock {\em IEEE Control systems magazine}, vol. 28, no. 3, pp. 36--73,
  2008.

\bibitem{wu2021new}
X.~Wu, S.~Magnusson, and M.~Johansson,
\newblock ``A new family of feasible methods for distributed resource
  allocation,''
\newblock in {\em IEEE Conference on Decision and Control}, 2021, pp.
  3355--3360.

\bibitem{bertsekas2003convex}
D~Bertsekas, A~Nedic, and A~Ozdaglar,
\newblock {\em Convex Analysis and Optimization},
\newblock Athena Scientific, Belmont, MA, 2003.

\bibitem{slp_book}
D.~Jurafsky and J.~H. Martin,
\newblock {\em Speech and Language Processing},
\newblock Prentice Hall, 2020.

\bibitem{nesterov1998introductory}
Y.~Nesterov,
\newblock ``Introductory lectures on convex programming, {volume I}: Basic
  course,''
\newblock {\em Lecture notes}, vol. 3, no. 4, pp. 5, 1998.

\bibitem{chen2016distributed}
G.~Chen, J.~Ren, and E.~N. Feng,
\newblock ``Distributed finite-time economic dispatch of a network of energy
  resources,''
\newblock {\em IEEE Transactions on Smart Grid}, vol. 8, no. 2, pp. 822--832,
  2016.

\bibitem{ccta}
M.~Doostmohammadian, A.~Aghasi, M.~Vrakopoulou, and T.~Charalambous,
\newblock ``1st-order dynamics on nonlinear agents for resource allocation over
  uniformly-connected networks,''
\newblock in {\em IEEE Conference on Control Technology and Applications}.
  IEEE, 2022, pp. 1184--1189.

\bibitem{scl}
M.~Doostmohammadian, A.~Aghasi, M.~Vrakopoulou, H.~R. Rabiee, U.~A. Khan, and
  T.~Charalambous,
\newblock ``Distributed delay-tolerant strategies for equality-constraint
  sum-preserving resource allocation,''
\newblock {\em Systems \& Control Letters}, vol. 182, pp. 105657, 2023.

\bibitem{paganini_book}
G.~E. Dullereud and F.~Paganini,
\newblock {\em A course in robust control theory: a convex approach},
\newblock Springer, 1999.

\bibitem{wang2022byzantine}
R.~Wang, Y.~Liu, and Q.~Ling,
\newblock ``Byzantine-resilient resource allocation over decentralized
  networks,''
\newblock {\em IEEE Transactions on Signal Processing}, vol. 70, pp.
  4711--4726, 2022.

\end{thebibliography}
\end{document}